\documentclass[aip]{revtex4-1}
\usepackage{amsfonts,amsthm,amsmath,graphicx,hyperref}

\newtheorem{theorem}{Theorem}[section] 
\newtheorem{conj}[theorem]{Conjecture}
\newtheorem{cor}[theorem]{Corollary}
\newtheorem{lemma}[theorem]{Lemma}
\newtheorem{prop}[theorem]{Proposition}
\theoremstyle{definition} 
\newtheorem{definition}[theorem]{Definition}
\theoremstyle{remark} 
\newtheorem{remark}[theorem]{Remark}

\DeclareMathOperator{\tr}{Tr} 
\DeclareMathOperator{\rank}{rank} 
\DeclareMathOperator{\ebr}{ebr}
\DeclareMathOperator{\Ad}{Ad}

\newcommand{\abs}[1]{\lvert#1\rvert}
\newcommand{\bb}[1]{\mathbb{#1}}
\newcommand{\cl}[1]{\mathcal{#1}}
\newcommand{\ff}[1]{\mathfrak{#1}}
\newcommand{\inner}[2]{\left\langle {#1},{#2} \right\rangle}

\begin{document}

\title{Entanglement breaking rank and the existence of SIC POVMs}

\author{Satish K. Pandey}
\email{satish.pandey@uwaterloo.ca}
\affiliation{Department of Pure Mathematics, University of Waterloo}
\affiliation{Institute for Quantum Computing, University of Waterloo}

\author{Vern I. Paulsen}
\email{vpaulsen@uwaterloo.ca}
\affiliation{Department of Pure Mathematics, University of Waterloo}
\affiliation{Institute for Quantum Computing, University of Waterloo}

\author{Jitendra Prakash}
\email{jprakash@uwaterloo.ca}
\affiliation{Department of Pure Mathematics, University of Waterloo}
\affiliation{Institute for Quantum Computing, University of Waterloo}

\author{Mizanur Rahaman}
\email{mizanurr@goa.bits-pilani.ac.in}
\affiliation{Department of Mathematics, Birla Institute of Technology and Science, Pilani – Goa Campus}

\date{\today}

\begin{abstract} 
We introduce and study the entanglement breaking rank of an entanglement breaking channel. We show that the entanglement breaking rank of the channel $\ff Z:\bb M_d \to \bb M_d$ defined by $\ff Z(X) = \frac{1}{d+1}(X+\tr(X)\bb I_d)$ is $d^2$ if and only if there exists a symmetric informationally-complete POVM in dimension $d$.
\end{abstract}


\maketitle


\maketitle 

\section{Introduction}
The study of separable states is an important topic in quantum information theory and helps to shed light on the nature of entanglement. Recall that, a state $\rho\in \bb M_m\otimes \bb M_n$ is called \emph{separable} if it can be written as a finite convex combination $\rho=\sum_{i}\lambda_i\sigma_i \otimes \delta_i$, where $\sigma_i$ and $\delta_i$ are pure states. In general, a separable state can have many such representations \cite{HJW}. It is then natural to introduce the notion of \emph{optimal ensemble cardinality} \cite{DTT} or \emph{length} \cite{CD} of a separable state by defining the length to be the minimum number $\ell(\rho)$ of pure states $\sigma_i \otimes \delta_i$ required to write the separable state $\rho$ as such a convex combination. In Ref. \cite{DTT} it was shown that, in general, the inequality $\rank(\rho)\leq \ell(\rho)$ can be strict for certain classes of separable states.

A \emph{quantum channel} $\Phi:\bb M_d\to\bb M_m$ is a completely positive and trace preserving linear map. It is well known that every quantum channel admits a Choi-Kraus decomposition: \begin{align}\label{eq:Choi-Kraus-of-EB}
\Phi(X) = \sum_{k=1}^K R_kXR_k^*, \qquad X\in \bb M_d,
\end{align} where the $R_k$'s are in $\bb M_{m,d}$, and satisfy the relation $\displaystyle \sum_{k=1}^K R_k^*R_k=\bb I_d$. There is an isomorphism between bipartite quantum states and quantum channels given by \begin{align*}
\Phi \mapsto \sum_{i,j} E_{i,j}\otimes \Phi(E_{i,j}),
\end{align*} (up to a suitable normalization), where $E_{i,j}$'s are the canonical matrix units. This isomorphism carries the convex set of separable states onto a class of quantum channels, known as the \emph{entanglement breaking channels}. An equivalent criterion for a channel to be entanglement breaking \cite{HSR} is that the Choi-Kraus operators $R_k$'s appearing in Equation \eqref{eq:Choi-Kraus-of-EB}, can be chosen to have rank one. The paper \cite{HSR} contains several other important characterizations of entanglement breaking channels, which we shall recall later. Since Choi-Kraus representations are not unique, we define the \emph{entanglement breaking rank} of an entanglement breaking channel $\Phi$ to be the minimum $K$ required in such an expression \eqref{eq:Choi-Kraus-of-EB} with $R_k$'s being rank one. We denote this quantity by $\ebr(\Phi)$. It is readily seen that in the correspondence between channels and bipartite states, that the entanglement breaking rank is the length of the corresponding separable bipartite state.

The existence of \emph{symmetric informationally-complete positive operator-valued measures} (SIC POVM) in an arbitrary dimension $d$ is an open problem and is an active area of research \cite{Busch,DPS,LS,Hoggar,RBSC,GA}. This problem has connections to frame theory and the theory of $t$-designs \cite{RBSC}, and has been verified numerically for many dimensions \cite{Appleby,Grassl,SG}. This problem is equivalent to   Zauner's conjecture \cite{Zauner1, Zauner} about the existence of $d^2$ equiangular lines in $\bb C^d$, which we refer to as \emph{Zauner's weak conjecture}. Stated in terms of equiangular lines, this problem asks about the existence of $d^2$ unit vectors $\{v_i:1\leq i\leq d^2\}\subseteq \bb C^d$ such that $\abs{\inner{v_i}{v_j}}^2 = \frac{1}{d+1}$ for all $i\neq j$. We refer the reader to Ref. \cite{FHS} for a survey on this subject. \emph{Zauner's strong conjecture} posits a particular form for these vectors.

It can be shown (see Ref.\cite{Caves99}) that if such a set of vectors $\{v_i\}_{i=1}^{d^2}$ exists in a given dimension $d$, then the rank one projections $P_i$ onto the span of $v_i$ satisfy: \begin{align}\label{Z-formula}
\frac{1}{d} \sum_{i=1}^{d^2} P_i XP_i = \frac{1}{d+1}(X+\tr(X)\bb I_d), \qquad X\in \bb M_d.
\end{align} Denoting this channel by $\ff Z$, it is independently known to be entanglement breaking, and it can be shown that its entanglement breaking rank is at least $d^2$ (see Ref.\cite{LH,HH}). Thus, if a SIC POVM exists in dimension $d$, then the entanglement breaking rank of the channel $\ff Z$ is $d^2$. 

We prove the converse: if $\ebr(\ff Z)= d^2$, then there exists a SIC POVM in dimension $d$. This leads to the conclusion (Corollary ~\ref{ebr-Zauner}) that for $d\geq 2$,  $\ebr(\ff Z)=d^2$ if and only if a SIC POVM exists in dimension $d$.

This result leads us to seek, and find, other channels with the property that a statement about their entanglement breaking rank is equivalent to Zauner's weak conjecture. In particular, we show that the channel, \begin{align*}
T\circ \ff Z(X) = \frac{1}{d+1}(X^T+\tr(X)\bb I_d), \qquad X\in \bb M_d,
\end{align*} where $T$ denotes the transpose map, has entanglement breaking rank equal to $d^2$ if and only if Zauner's weak conjecture is true in dimension $d$. Often the channels of the form 
\[X\mapsto \frac{1}{d\pm1}(\tr(X)\bb I_d)\pm X^T)\]
are called Werner-Holevo channels (see Ref. \cite{Watr}, Example 3.36) and they have proven to be useful in many problems in quantum information theory.

The channel $\ff Z$ is a particular convex combination of the \emph{identity channel} $\cl I_d:\bb M_d\to\bb M_d$ defined by $\cl I_d(X)=X$ for all $X\in \bb M_d$, and the \emph{completely depolarizing channel} $\Psi_d:\bb M_d\to \bb M_d$ defined by $\Psi_d(X)=\frac{1}{d}\tr(X)\bb I_d$ for all $X\in \bb M_d$. Namely, \begin{align*}
\ff Z = \frac{1}{d+1}\cl I_d + \frac{d}{d+1}\Psi_d.
\end{align*} More generally, for $t\in \bb R$, it is known  that the map $\Phi_t = t\cl I_d+(1-t)\Psi_d$ is an entanglement breaking channel if and only if $\frac{-1}{d^2-1}\leq t\leq \frac{1}{d+1}$ (see Theorem \ref{thm:conv-comb} and references therein). In this context, we prove some perturbative results that would imply Zauner's conjecture. We prove that \emph{ebr} is lower semicontinuous. Hence, if there are channels arbitrarily near to $\ff Z$ with \emph{ebr} bounded by $d^2$, then Zauner's weak conjecture is true. 
We conjecture that, for all $t$ with $0 \leq t \leq \frac{1}{d+1}$, the entanglement breaking rank of each of these channels is $d^2$. We do not know if our conjecture is stronger than Zauner's conjecture or if it is a consequence of Zauner's conjecture.  
We verify this stronger conjecture in dimensions $d=2$ and $d=3$. 

Since determining if $\ebr(\ff Z)=d^2$ is equivalent to  Zauner's weak conjecture, it is interesting to know what sorts of bounds we can obtain on $\ebr(\ff Z)$, that do not rely on Zauner's weak conjecture.
In addition to the SIC POVM existence problem, there is also the existence problem of \emph{mutually unbiased bases} (MUB) in dimension $d$ (see Ref. \cite{DEBZ,Kantor}). In Section \ref{sec:MUB}, we show that the existence of $d+1$ MUB in dimension $d$ implies $\ebr(\ff Z) \le d(d+1)$. 

For the purpose of disambiguation, we remark that this is a revised version of the paper titled, ``Entanglement breaking rank".

\section{Preliminaries}
Let $d\in \bb N$. We shall let $\bb C^d$ denote the Hilbert space of $d$-tuples of complex numbers with the usual inner product. We shall denote the space of $d \times d$ complex matrices by $\mathbb{M}_d$. The trace of a matrix $A=[a_{i,j}]\in \mathbb{M}_d$ is defined by $\tr(A)=\sum_{i=1}^d a_{i,i}$. We let $\bb I_d$ denote the identity matrix in $\bb M_d$. The space of all linear maps from $\bb M_d$ to $\bb M_m$ is denoted by $\cl L(\bb M_d,\bb M_m)$.

We review some of the definitions and known results that we shall be needing. We begin with stating the SIC POVM existence problem in terms of equiangular vectors in $\bb C^d$.

\begin{definition}\label{def:SIC-POVM}
A set of $d^2$ matrices $\{R_i:1\leq i\leq d^2\}\subseteq \bb M_d$ is called a \emph{symmetric informationally-complete positive operator-valued measure} (SIC POVM) if, \begin{enumerate}
\item[(a)] it forms a \emph{positive operator-valued measure} (POVM), that is, each $R_i$ is positive semi-definite (written $R_i\geq 0$) and $\sum_{i=1}^{d^2}R_i=\bb I_d$,
\item[(b)] it is \emph{informationally-complete}, that is, $\text{span}\{R_i:1\leq i\leq d^2\}=\bb M_d$,
\item[(c)] it is \emph{symmetric}, that is, $\tr(R_i^2)=\lambda$ for all $1\leq i\leq d^2$, and $\tr(R_iR_j)=\mu$ for all $i\neq j$, for some constants $\lambda$ and $\mu$, and,
\item[(d)] $\rank(R_i)=1$ for all $1\leq i\leq d^2$.
\end{enumerate}
\end{definition}

The following proposition gives a characterization of a SIC POVM in terms of unit vectors.

\begin{prop}\label{SIC-iff-equiangular}
A set of matrices $\{R_i\}_{i=1}^{d^2}\subseteq \bb M_d$ is a SIC POVM if and only if there exist $d^2$  unit vectors $\{w_i\}_{i=1}^{d^2}\subseteq \bb C^d$ such that $R_i=\frac{1}{d}w_iw_i^*$ for each $1\leq i\leq d^2$, and $|\langle w_i,w_j\rangle|^2=\frac{1}{d+1}$ for all $i\neq j$.
\end{prop}

Because of Proposition \ref{SIC-iff-equiangular}, we will say that a set $\{w_i\}_{i=1}^{d^2}\subseteq \bb C^d$ of $d^2$ unit vectors generates a SIC POVM if $|\inner{w_i}{w_j}|^2 = \frac{1}{d+1}$ for all $i\neq j$. 

The constant $\frac{1}{d+1}$ appearing in Proposition \ref{SIC-iff-equiangular} can be derived independently.

\begin{prop}\label{prop:constant-angle}
If $\{w_i\}_{i=1}^{d^2}\subseteq \bb C^d$ is a set of unit vectors such that $|\inner{w_i}{w_j}|^2=c\neq 1$ for all $i\neq j$, then $c=\frac{1}{d+1}$.
\end{prop}
 
We now state the SIC POVM existence conjecture in terms of equiangular vectors, which we shall call \emph{Zauner's weak conjecture}.

\begin{conj}[Zauner's weak conjecture]\label{conj-zauner}
For each positive integer $d\geq 2$, there exist $d^2$ unit vectors $\{w_i\}_{i=1}^{d^2}$ in $\bb C^d$ such that $|\langle w_i,w_j\rangle|^2=\frac{1}{d+1}$ for all $i\neq j$.
\end{conj}
The existence of a SIC POVM for each dimension $d\geq 2$ is still an open problem, however in some specific dimensions $(d=1-21,24, 28, 30, 31, 35 \ \text{etc.})$ concrete solutions of SIC sets have been found. See Ref. \cite{Scott2017} for current progress in finding these solutions.
Many of the numerical solutions to Zauner's weak conjecture have been found assuming a covariance property. We shall focus on a particular group covariance. 

\begin{definition}
For a positive integer $d$, let $\mathbb{Z}_d = \{0, 1, 2, \cdots, d-1 \}$ be the cyclic group of order $d$. Let $\omega=\exp\left(\frac{2\pi i}{d}\right)$ be a $d$-th root of unity. Define the following two $d\times d$ matrices \begin{align*}
U=\begin{bmatrix}
0 & 0 & \cdots & 1  \\
1 & 0 & \cdots & 0 \\
0 & \ddots & 0 & 0 \\
0 & \cdots & 1 & 0 
\end{bmatrix}, \qquad V=\begin{bmatrix}
1 & 0 & \cdots & 0  \\
0 & \omega & \cdots & 0 \\
0 & 0 & \ddots & 0 \\
0 & 0 & \cdots & \omega^{d-1} 
\end{bmatrix}.
\end{align*} Thus, $U$ is the (forward) cyclic shift, and $V$ is the diagonal matrix with all the $d$-th roots of unity on the diagonal. For any $(i,j)\in \mathbb{Z}_d\times\mathbb{Z}_d$, define the \emph{discrete Weyl matrices} as $W_{i,j}=U^{i}V^{j}$. (For properties of the discrete Weyl matrices, we refer the reader to Ref. \cite{Watr} , {Section 4.1.2}.)
\end{definition}

We now write down Zauner's strong conjecture. It is not known whether Zauner's weak conjecture implies this stronger conjecture.

\begin{conj}[Zauner's strong conjecture]\label{conj-strng-zauner}
For each positive integer $d\geq 2$, there exists a unit vector $w\in \mathbb{C}^d$ such that the $d^2$ vectors, \begin{align*}
\{W_{i,j}w: (i,j)\in \mathbb{Z}_d\times\mathbb{Z}_d\} \subseteq \bb C^d,
\end{align*} where $W_{i,j}$ are the discrete Weyl matrices, generate a SIC POVM. In this case, $w$ is called a \emph{fiducial vector}.
\end{conj}

We now review some of the standard results from the theory of quantum channels. The proofs may be found in Ref. \cite{Choi,Watr}. Recall that for a linear map $\Phi:\bb M_d\to \bb M_m$, there is a unique bipartite element $C_{\Phi} = [\Phi(E_{i,j})]_{i,j=1}^d$ in $\mathbb{C}^{dm}\otimes\mathbb{C}^{dm}$. Moreover,  $\Phi$ is a completely positive (CP) linear map if and only if  $C_{\Phi}$ is a positive semidefinite matrix. The matrix $C_\Phi$ is called the Choi matrix of the map $\Phi$. We will denote the rank of this Choi matrix as $\mathrm{cr}(\Phi)$. As described in the introduction, every quantum channel admits a Choi-Kraus representation given in the equation \ref{eq:Choi-Kraus-of-EB}. Equivalently the Choi rank of $\Phi$, $\mathrm{cr}(\Phi)$, is the minimum $K$ in the equation \ref{eq:Choi-Kraus-of-EB} among all possible Choi-Kraus representations.

We now recall the definition and key facts about entanglement breaking channels.
\begin{definition}
A quantum channel $\Phi:\bb M_d\to \bb M_m$ is called \emph{entanglement breaking} if for all $n\in \bb N$, the $n$-th amplification map \begin{align*}
\Phi\otimes \cl I_n : \bb M_d\otimes \bb M_n \to \bb M_m\otimes \bb M_n
\end{align*} maps all states (entangled or not) in $\bb M_d\otimes \bb M_n$ to separable states in $\bb M_m\otimes \bb M_n$.
\end{definition}

\begin{theorem}\cite{HSR}\label{EB-equivalences}
Let $\Phi:\bb M_d\to \bb M_m$ be a quantum channel. Then the following statements are equivalent. \begin{enumerate}
\item[(a)] The map $\Phi$ is entanglement breaking.
\item[(b)] The Choi-matrix $C_{\Phi}$ is a separable state.
\item[(c)] There is a Choi-Kraus representation \ref{eq:Choi-Kraus-of-EB} of $\Phi$, such that each $R_k$ has rank one.
\item[(d)] The linear map $\Phi{^\prime}\circ\Phi:\bb M_d\to \bb M_k$ is completely positive for every positive map $\Phi^{\prime}:\bb M_m\to \bb M_k$.
\item[(e)] The linear map $\Phi\circ\Phi{^\prime}:\bb M_k\to \bb M_m$ is completely positive for every positive map $\Phi^{\prime}:\bb M_k\to \bb M_d$.
\end{enumerate}

Moreover, the set of all entanglement breaking maps $\Phi:\bb M_d\to \bb M_m$ forms a convex set in $\cl L(\bb M_d, \bb M_m)$.
\end{theorem}
As mentioned in the introduction, the minimum number of rank one Kraus operators needed to express an entanglement breaking map $\Phi$ is called the \emph{entanglement breaking rank} ($ebr(\Phi)$) of $\Phi$. Evidently one has
\begin{align}\label{eq:choi-ebr-rank}
d\leq \mathrm{cr}(\Phi)\leq \ebr(\Phi).
\end{align} The first inequality follows from the fact that an entanglement breaking channel cannot have less than $d$ Choi-Kraus operators in its Choi-Kraus representation (see Theorem 6 in Ref. \cite{HSR}).

For $t\in \bb R$, we define the map $\Phi_t:\bb M_d\to \bb M_d$ to be $\Phi_t = t\cl I_d+(1-t)\Psi_d$. When $t=\frac{1}{d+1}$, we get the channel $\ff Z$. We finish this section by mentioning the following result which states the values of $t$ for which the map $\Phi_t$ becomes a channel, and an entanglement breaking channel. 

\begin{theorem}\cite{LH,HH,FFMV}\label{thm:conv-comb}
For $t\in \bb R$, consider the map $\Phi_t$ on $\bb M_d$ and let $T:\bb M_d\to \bb M_d$ be the transpose map. Then, \begin{enumerate}
\item[(a)] $\Phi_t$ is a channel (called a \emph{depolarizing channel}) if and only if $t\in \left[\frac{-1}{d^2-1},1\right]$,
\item[(b)] $T\circ\Phi_t$ is a channel (called a \emph{transpose depolarizing channel}) if and only if $t\in \left[\frac{-1}{d-1},\frac{1}{d+1} \right]$,
\item[(c)] $\Phi_t$ and $T\circ\Phi_t$ are entanglement breaking if and only if $t\in \left[\frac{-1}{d^2-1},\frac{1}{d+1}\right]$.
\end{enumerate}
\end{theorem}

When $t\in \left(\frac{-1}{d^2-1},1\right)$, it is easy to see that the Choi-rank of $\Phi_t$ is $d^2$. Indeed, when $t=0$, $\Phi_t$ is the completely depolarizing map $\Psi_d$, and since $C_{\Psi_d} = [\Psi_d(E_{i,j})]=\frac{1}{d}\bb I_{d^2}$, it follows that $\mathrm{cr}(\Psi_d)=d^2$. When $t\neq 0$, but is in the aforementioned interval, the Choi-matrix of $\Phi_t$ is $C_{\Phi_t} = t[E_{i,j}]+\frac{1-t}{d}\bb I_{d^2}$. Let $\xi\in \bb C^{d^2}$. If $C_{\Phi_t}(\xi)=0$, then $[E_{i,j}]\xi = -\frac{1-t}{td}\xi$. But the only eigenvalues of $[E_{i,j}]$ are $0$ and $d$, which implies that $\xi=0$. Thus $\mathrm{cr}(\Phi_t) = \rank(C_{\Phi_t})=d^2$.

\section{Entanglement Breaking Rank}
The main goal of this section is to establish the fact that determining the entanglement breaking rank of the channel $\ff Z=\frac{1}{d+1}\cl I_d + \frac{d}{d+1}\Psi_d$ to be $d^2$ is equivalent  to the existence problem of SIC POVMs. 
By Proposition \ref{thm:conv-comb}, it follows that the quantum channel $\ff Z$ is entanglement breaking, and hence it has a Choi-Kraus representation consisting of rank one operators. Zauner's weak conjecture can then be related to the problem of obtaining a \emph{minimal} Choi-Kraus representation of the quantum channel $\ff Z$ consisting of rank one operators. First we establish a weaker proposition.

\begin{prop}\label{Channel-to-Zauner}
Zauner's weak conjecture is true if and only if for each positive integer $d\geq 2$, the Choi-Kraus representation of $\ff Z$ consists of $d^2$ \emph{positive} rank one operators.
\end{prop}

\begin{proof}
The forward implication is a known result \cite{Caves99}. However, we include a proof for the sake of completeness. 

Suppose that Zauner's weak conjecture is true. Then for each $d\geq 2$, there exist $d^2$ unit vectors $\{w_i\}_{i=1}^{d^2}$ in $\bb C^d$ such that $|\inner{w_i}{w_j}|^2=\frac{1}{d+1}$ for all $i\neq j$. By Proposition \ref{SIC-iff-equiangular}, the set $\{R_i\}_{i=1}^{d^2}\subseteq \bb M_d$, where $R_i=\frac{1}{d}w_iw_i^*$, forms a SIC POVM. Note that $R_i^2 = \frac{1}{d}R_i$ for every $1\leq i\leq d^2$. Define a map $\Phi:\bb M_d\to \bb M_d$ by 
\begin{align*}
\Phi(X) = d \sum_{i=1}^{d^2}R_iXR_i^*, \qquad X\in \bb M_d.
\end{align*}
Clearly, $\Phi$ is completely positive, and the following shows that it is unital and trace preserving: \begin{align*}
d\sum_{i=1}^{d^2}R_i^*R_i = d\sum_{i=1}^{d^2}R_i^2=d\sum_{i=1}^{d^2}\frac{1}{d}R_i = \sum_{i=1}^{d^2}R_i = \bb I_d.
\end{align*}
Thus, $\Phi$ is a unital quantum channel. Next we observe that \begin{align*}
\Phi(R_j) &= d\left(R_j^3+\sum_{\substack{1\leq i\leq d^2 \\ i\neq j}} R_iR_jR_i \right) = d\left(\frac{1}{d^2}R_j+\frac{1}{d^2(d+1)} \sum_{\substack{1\leq i\leq d^2 \\ i\neq j}} R_i \right) \\
&= d\left(\frac{1}{d^2}R_j+ \frac{1}{d^2(d+1)} (\bb I_d-R_j) \right) = \frac{1}{d+1}\left(R_j+\frac{1}{d} \bb I_d\right) = \ff Z(R_j).
\end{align*} Since the set $\{R_j\}_{j=1}^{d^2}$ is informationally complete, it spans $\bb M_d$, and therefore it follows that $\Phi=\ff Z$, which yields the desired Choi-Kraus representation of $\ff Z$ with $B_i=\sqrt dR_i$ for every $1\leq i\leq d^2$.

Conversely, suppose $\ff Z$ has a Choi-Kraus representation given by $\ff Z(X)=\sum_{i=1}^{d^2}B_iXB_i^*$, where each $B_i\in \bb M_d$ is a rank one positive operator. Then for every $1\leq i\leq d^2$, $B_i=v_iv_i^*$ for some $v_i\in \bb C^d$. Since the channel $\ff Z$ is unital, we have 
\begin{align*}
\bb I_d=\sum_{i=1}^{d^2}B_i^2 = \sum_{i=1}^{d^2}\|v_i\|^2B_i.
\end{align*} Using this, and using $\ff Z=\frac{1}{d+1}\cl I_d+\frac{d}{d+1}\Psi_d$, on one hand we get \begin{align*}
\ff Z(B_j) &= \frac{1}{d+1}(B_j+\|v_j\|^2\bb I_d) = \frac{1}{d+1}\left(B_j+\sum_{i=1}^{d^2}\|v_j\|^2\|v_i\|^2 B_i\right) \\
&= \frac{1}{d+1}\left((1+\|v_j\|^4)B_j+ \sum_{\substack{1\leq i\leq d^2 \\ i\neq j}}\|v_i\|^2\|v_j\|^2 B_i \right);
\end{align*} and on the other hand, using Choi-Kraus representation of $\ff Z$, we get, \begin{align*}
\ff Z(B_j) = \sum_{i=1}^{d^2}B_iB_jB_i = B_j^3 + \sum_{\substack{1\leq i\leq d^2 \\ i\neq j}} B_iB_jB_i = \|v_j\|^4B_j  + \sum_{\substack{1\leq i\leq d^2 \\ i\neq j}} |\inner{v_i}{v_j}|^2 B_i.
\end{align*} Comparing $\ff Z(B_j)$ obtained in two ways, we get \begin{align*}
\left(\frac{1+\|v_j\|^4}{d+1} - \|v_j\|^4 \right) B_j + \sum_{\substack{1\leq i\leq j\\i\neq j}} \left( 
\frac{\|v_i\|^2\|v_j\|^2}{d+1} - |\inner{v_i}{v_j}|^2\right) B_i = 0.
\end{align*} Since $\{B_i\}_{i=1}^{d^2}$ is a linearly independent set (because number of Choi-Kraus operators equals the Choi-rank), we must have \begin{align*}
\frac{1+\|v_j\|^4}{d+1} - \|v_j\|^4 = 0,\qquad \frac{\|v_i\|^2\|v_j\|^2}{d+1}-|\inner{v_i}{v_j}|^2 = 0, \forall \;i\neq j.
\end{align*} The first one yields, $\|v_j\|^4=\frac{1}{d}$, which is constant for all $1\leq j\leq d^2$, and using this the second one yields $|\inner{v_i}{v_j}|^2 = \frac{1}{d(d+1)}$, for all $i\neq j$. Then it is easy to see that the normalized vectors $w_i:=\frac{v_i}{\|v_i\|}$ satisfy $|\inner{w_i}{w_j}|^2=\frac{1}{d+1}$ for all $i\neq j$, so that Zauner's weak conjecture holds.
\end{proof}

The following result shows that the positivity condition on $B_i$ in the Proposition \ref{Channel-to-Zauner} is redundant and can be dropped. 

\begin{theorem}\label{Channel-to-Zauner-2}
Zauner's weak conjecture is true if and only if for each positive integer $d\geq 2$, the Choi-Kraus representation of $\ff Z$ consists of $d^2$ rank one operators.
\end{theorem}

\begin{proof}
It suffices to prove the backward implication. 
Suppose that $\ff Z$ has a Choi-Kraus representation given by $\ff Z(X)=\sum_{i=1}^{d^2}B_iXB_i^*$, where each $B_i\in \bb M_d$ is a rank one operator. Then for every $1\leq i\leq d^2$, $B_i=x_iy_i^*$ for some vectors $x_i, y_i\in \bb C^d$. Without loss of generality, we may assume that each $y_i$ is a unit vector. If we show that for every $1\leq i\leq d^2$, $x_i=\lambda_iy_i$ for some $\lambda_i>0$, then the Choi-Kraus representation of $\ff Z$ reduces to a form which consists of rank one positive operators. Then the result follows from Proposition \ref{Channel-to-Zauner}.

We first prove that the set $\{x_ix_i^*\}_{i=1}^{d^2}\subseteq \bb M_d$ is a basis for $\bb M_d$.
To do this we note that $\ff Z$ is unital and hence we have 
\begin{align}\label{eq:unital-phi}
\bb I_d = \ff Z(\bb I_d) = \sum_{i=1}^{d^2}B_iB_i^* = \sum_{i=1}^{d^2} x_iy_i^*y_ix_i^* = \sum_{i=1}^{d^2} x_ix_i^*.
\end{align} 
Using this equation, on one hand we get \begin{align*}
\ff Z(B_j) = \frac{1}{d+1}\left(B_j+\tr(B_j)\bb I_d \right) = \frac{1}{d+1}\left(B_j+\inner{x_j}{y_j}\left(\sum_{i=1}^{d^2}x_ix_i^* \right) \right),
\end{align*} 
and on the other hand we get 
\begin{align*}
\ff Z(B_j) = \sum_{i=1}^{d^2} B_iB_jB_i^* = \sum_{i=1}^{d^2} (x_iy_i^*)(x_jy_j^*)(x_iy_i^*)^* = \sum_{i=1}^{d^2} \inner{x_j}{y_i}\inner{y_i}{y_j} x_ix_i^*,
\end{align*}
for $1\leq j \leq d^2$.
Comparing $\ff Z(B_j)$ obtained in the above two ways, we get 
\begin{align*}
\frac{1}{d+1}\left(B_j+\inner{x_j}{y_j}\left(\sum_{i=1}^{d^2}x_ix_i^* \right) \right) =  \sum_{i=1}^{d^2} \inner{x_j}{y_i}\inner{y_i}{y_j} x_ix_i^*,
\end{align*} which implies \begin{align*}
B_j = \sum_{i=1}^{d^2} \left((d+1)\inner{x_j}{y_i}\inner{y_i}{y_j}-\inner{x_j}{y_j} \right) x_ix_i^*.
\end{align*}
This shows that the set $\{x_ix_i^*\}_{i=1}^{d^2}$ spans $\{B_j\}_{j=1}^{d^2}$. But since the Choi-rank of $\ff Z$ is $d^2$, $\{B_j\}_{j=1}^{d^2}$ is a basis for $\bb M_d$. Consequently, the set $\{x_ix_i^*\}_{i=1}^{d^2}$ is a minimal spanning set, and hence a basis for $\bb M_d$.

We now show that for each $j$, $x_j$ is a scalar multiple of $y_j$. Using Equation \eqref{eq:unital-phi}, for each $1\leq j \leq d^2$, on one hand we get
\begin{align*}
\ff Z(x_jx_j^*) &= \frac{1}{d+1}\left(x_jx_j^*+\tr(x_jx_j^*)\bb I_d \right) = \frac{1}{d+1}\left(x_jx_j^* + \|x_j\|^2\left(\sum_{i=1}^{d^2}x_ix_i^* \right) \right),
\end{align*} and on the other hand, we get \begin{align*}
\ff Z(x_jx_j^*) &= \sum_{i=1}^{d^2} B_i(x_jx_j^*)B_i^* = \sum_{i=1}^{d^2}x_iy_i^*x_jx_j^*y_ix_i^* = \sum_{i=1}^{d^2} |\inner{y_i}{x_j}|^2 x_ix_i^*.
\end{align*} Comparing the two expressions of $\ff Z(x_jx_j^*)$, we get \begin{align*}
\frac{1}{d+1}\left(x_jx_j^* + \|x_j\|^2\left(\sum_{i=1}^{d^2}x_ix_i^* \right) \right) &= \sum_{i=1}^{d^2} |\inner{y_i}{x_j}|^2 x_ix_i^*.
\end{align*} Because of linear independence of the set $\{x_ix_i^*\}_{i=1}^{d^2}$, comparing the coefficients of $x_ix_i^*$, it follows that \begin{align}
(d+1)|\inner{y_j}{x_j}|^2 &= 1+\|x_j\|^2, \qquad \forall j, \label{eq:zero-coeff-1}\\
(d+1)|\inner{y_i}{x_j}|^2 &= \|x_j\|^2, \qquad \forall i\neq j.\label{eq:zero-coeff-2}
\end{align} 
Using the Cauchy-Schwarz inequality in Equation \eqref{eq:zero-coeff-1}, we obtain
\begin{align*}
1+\|x_j\|^2 = (d+1)|\inner{y_j}{x_j}|^2 \leq (d+1)\|y_j\|^2\|x_j\|^2 = (d+1)\|x_j\|^2,
\end{align*} which implies $\frac{1}{d}\leq \|x_j\|^2$, and so $\sum_{j=1}^{d^2}\|x_j\|^2\geq d$. Moreover, taking the trace of Equation \eqref{eq:unital-phi}, we get $\sum_{i=1}^{d^2}\|x_i\|^2=d$. This implies that equality holds in the Cauchy-Schwarz inequality above, and therefore, $x_j=\lambda_j y_j$, for some scalar $\lambda_j$, $1\leq j\leq d^2$. 

Since there is no loss of generality in assuming that $\lambda_j>0$, it follows that each $B_j$ is indeed a positive rank one operator, which proves the theorem. \end{proof}

The following corollary is an immediate consequence of Theorem \ref{Channel-to-Zauner-2} along with the fact that $\mathrm{cr}(\ff Z)=d^2$.

\begin{cor}\label{ebr-Zauner}
Zauner's weak conjecture holds if and only if $\ebr(\ff Z)=d^2$ for all $d\geq 2$.
\end{cor}

Interestingly, for $t< \frac{1}{d+1}$, the channel $\Phi_t = t\cl I_d+(1-t)\Psi_d$ cannot have a Choi-Kraus representation with $d^2$ \emph{positive} rank one Choi-Kraus operators, which we prove next.

\begin{prop}
Fix an integer $d \ge 2$ and let $t\in \left(\frac{-1}{d^2-1},\frac{1}{d+1} \right)$. Let $\Phi_t:\bb M_d\to \bb M_d$ be the quantum channel given by $\Phi_t = t\cl I_d+(1-t)\Psi_d$. Then $\Phi_t$ cannot have a Choi-Kraus representation, $\Phi_t(X)=\sum_{i=1}^{d^2}R_iXR_i^*$, where each $R_i$ is a positive rank one operator.
\end{prop}

\begin{proof}
We shall follow the arguments as in Proposition \ref{Channel-to-Zauner}. Suppose $\Phi_t$ has a Choi-Kraus representation, $\Phi_t(X)=\sum_{i=1}^{d^2}R_iXR_i^*$, where $R_i=v_iv_i^*$ for some vector $v_i\in \bb C^d$. Since $\Phi_t$ is unital, we have $\bb I_d=\sum_{i=1}^{d^2}\|v_i\|^2R_i$. Applying $\Phi_t$ to $R_j$, \begin{align*}
\Phi_t(R_j) = tR_j + \frac{1-t}{d}\tr(R_j)\bb I_d = \left(t+\frac{1-t}{d}\|v_j\|^4 \right)R_j + \sum_{\substack{1\leq i\leq d^2 \\ i\neq j}}\frac{1-t}{d}\|v_j\|^2\|v_i\|^2R_i,
\end{align*} and also \begin{align*}
\Phi_t(R_j) = \sum_{i=1}^{d^2}R_iR_jR_i^* =  \|v_j\|^4R_j  +\sum_{\substack{1\leq i\leq d^2\\i\neq j}}|\inner{v_i}{v_j}|^2R_i.
\end{align*} Comparing both the expressions of $\Phi_t(R_j)$ and using the linear independence of $R_i$'s, we get \begin{align*}
t+\frac{1-t}{d} \|v_j\|^4 = \|v_j\|^4, \qquad \frac{1-t}{d}\|v_j\|^2\|v_i\|^2 = |\inner{v_i}{v_j}|^2, \quad\forall i\neq j.
\end{align*} This implies \begin{align*}
\|v_j\|^4 = \frac{dt}{d+t-1}, \qquad |\inner{v_i}{v_j}|^2 = \frac{t(1-t)}{d+t-1}, \quad\forall i\neq j.
\end{align*} Let $w_i=\frac{v_i}{\|v_i\|}$. Then for all $i\neq j$, we have \begin{align*}
|\inner{w_i}{w_j}|^2 = \frac{1}{\|v_i\|^2\|v_j\|^2}|\inner{v_i}{v_j}|^2 = \frac{d+t-1}{dt}\frac{t(1-t)}{d+t-1} = \frac{1-t}{d}.
\end{align*} But by Proposition \ref{prop:constant-angle}, we must have $\frac{1-t}{d}=\frac{1}{1+d}$, which implies that $t=\frac{1}{d+1}$, which is a contradiction.
\end{proof}

In what follows, we seek other entanglement breaking channels with the property that a statement about their entanglement breaking rank is equivalent to Zauner's weak conjecture.

Let $\Phi:\bb M_d\to \bb M_d$ be an entanglement breaking map, $T:\bb M_d\to \bb M_d$ be the transpose map, and for a given unitary $U\in \bb M_d$ let $\Ad_U:\bb M_d\to \bb M_d$ be the map $\Ad_U(X)=UXU^*$. Since $T$ and $\Ad_U$ are both positive, it follows that $T\circ \Phi$ and $\Ad_U\circ \Phi$ are entanglement breaking. The following result determines the entanglement breaking rank of these two channels.

\begin{prop}\label{ebrtransposeinvariant}
Let $\Phi:\bb M_d\to \bb M_d$ be an entanglement breaking map. Then $\ebr(T\circ \Phi)=\ebr(\Phi)=\ebr(\Ad_U\circ \Phi)$.
\end{prop}

\begin{proof}
 Note that for any unitary $U$ and a matrix $R$, $UR$ has rank $1$ if and only if $R$ has rank $1$. Now this observation applied to the rank-$1$ Kraus operators of $\Phi$ immediately shows that $\ebr(\Phi)=\ebr(\Ad_U\circ \Phi)$.

Similarly, since $(v_jv_j^*)^T=\bar{v_j}v_j^T$, after composing the channel with the transpose map, the Kraus operators get transformed from $v_iw_j^*$ into $\bar{v_j}w_j^*$, which are still rank one. This shows that $\Phi$ is entanglement breaking if and only if $T\circ\Phi$ is. Now it is immediate that $\ebr(T\circ \Phi)=\ebr(\Phi).$ 
\end{proof}

\begin{cor}\label{equival-zauner} 
Let $d\geq 2$ and $\ff Z=\frac{1}{d+1}\cl I_d+\frac{d}{d+1}\Psi_d$. Then the following are equivalent. \begin{enumerate}
\item[(a)] Zauner's weak conjecture is true in dimension $d$.
\item[(b)] $\ebr(\ff Z) = d^2$.
\item[(c)] $\ebr(T\circ \ff Z) = d^2$, where $T$ is the transpose map.
\item[(d)] $\ebr(\Ad_U\circ \ff Z) = d^2$, for any unitary $U\in \bb M_d$.
\end{enumerate}
\end{cor}

For the next result we let $P_{d}: \bb C^d \otimes \bb C^d \to \bb C^d \otimes \bb C^d$ denote the projection onto the symmetric subspace, so that
\[ P_{d} ( e_k \otimes e_l) = \frac{1}{2}( e_k \otimes e_l + e_l \otimes e_k).\]

\begin{cor}\label{symproj-equival-zauner} Let $d \ge 2$. Then Zauner's weak conjecture is true in dimension $d$ if and only if there exist $2d^2$ vectors, $x_1,...,x_{d^2}, y_1, ..., y_{d^2}$ in $\bb C^d$ such that
\[ P_d = \sum_{j=1}^{d^2} (x_j \otimes y_j)(x_j \otimes y_j)^*.\]
Moreover, if $v_1,...,v_{d^2}$ is a set of $d^2$ unit vectors in $\bb C^d$ satisfying $| \langle v_i, v_j \rangle |^2 = \frac{1}{d+1}$ for all $i \ne j$, then
\[ P_d = \frac{d+1}{2d} \sum_{j=1}^{d^2} (v_j \otimes v_j) ( v_j \otimes v_j)^*.\]
\end{cor}
\begin{proof} Note that the Choi matrix $C$ of $T \circ \ff Z$ is $\frac{2}{d+1} P_d$ and that writing
\[ T \circ \ff Z(X) = \sum_{j=1}^{d^2} (x_jy_j^*)X(y_j x_j^*),\]
is equivalent to writing
\[ C= \sum_{j=1}^{d^2} (x_j \otimes \overline{y_j})(x_j \otimes \overline{y_j})^*,\]
from which the equivalence follows, since the scaling is irrelevant.

Next, given such a set of vectors, by \ref{Z-formula} we have that
\[\ff Z(X)= \frac{1}{d} \sum_{j=1}^{d^2} (v_j^*Xv_j) v_jv_j^*,\]
so that
\[ T \circ \ff Z(X) = \frac{1}{d} \sum_{j=1}^{d^2} (v_j^*Xv_j) \overline{v_j} \overline{v_j}^*= \frac{1}{d} \sum_{j=1}^{d^2} (\overline{v_j} v_j^*)X(\overline{v_j} v_j^*)^*.\]
From this it follows that
\[ P_d = \frac{d+1}{2} C = \frac{d+1}{2d} \sum_{j=1}^{d^2} (\overline{v_j} \otimes \overline{v_j})(\overline{v_j} \otimes \overline{v_j})^*\]
and the result follows by observing that the vectors $v_j$ are equiangular if and only if the vectors $\overline{v_j}$ are equiangular. \end{proof}

\begin{remark}
The channel $T\circ \ff Z$ serves as an example of a quantum channel whose entanglement breaking rank is strictly greater than its Choi-rank. To see this, first note that by Proposition~\ref{ebrtransposeinvariant}, we have \begin{align*}
\ebr(T\circ \ff Z)= \ebr(\ff Z) \ge \rank(C_{\ff Z})=d^2.
\end{align*} The Choi-rank of $T\circ \ff Z$ is $\frac{d(d+1)}{2}$, since that is the dimension of the symmetric subspace. Hence using the above inequality we get for any $d\geq 2$, \begin{align*}
\ebr(T\circ \ff Z)\geq d^2>\frac{d(d+1)}{2}=\mathrm{cr}(T\circ \ff Z).
\end{align*} This observation was first made in Ref. \cite{DTT}. \end{remark}

\section{Lower semicontinuity of entanglement breaking rank}
The following result shows that ebr is a lower semicontinuous function. Recall that if $X$ is a metric space, $x_0\in X$, and $f:X\to \bb R\cup \{\pm \infty\}$, then we say that $f$ is \emph{lower semicontinuous} at $x_0$ if and only if whenever $(x_n)_{n=1}^{\infty}$ is a sequence in $X$ which converges to $x_0$, we have $\liminf_{n\to \infty}f(x_n)\geq f(x_0)$.

\begin{prop}\label{prop:ebr-lower-semicts}
Let $(\Psi_n)_{n=1}^{\infty}$ be a sequence of entanglement breaking maps where $\Psi_n:\bb M_d\to \bb M_m$, and suppose that $\Psi_n \to \Psi$. Then \begin{align*}
\ebr(\Psi) \leq \liminf_n \ebr(\Psi_n).
\end{align*}
\end{prop}

\begin{proof}
It is enough to show that if $\ebr(\Psi_n) \le k$ for all $n\in \bb N$, then $\ebr(\Psi) \le k$. With that assumption, for each $n\in \bb N$, there exist rank one operators $\{R_{i,n}: 1 \le i \le k\}$ such that $\Psi_n(X) = \sum_{i=1}^k R_{i,n}XR_{i,n}^*$. Since $\Psi_n(\bb I_d)$ converges to $\Psi(\bb I_d)$ the sequence $(\Psi_n(\bb I_d))$ is bounded and so there is $c>0$ so that $\Psi_n(\bb I_d) \le c \bb I_d$ for all $n\in \bb N$. Hence, $R_{j,n}R_{j,n}^* \le \sum_{i=1}^k R_{i,n}R_{i,n}^* = \Psi_n(\bb I_d) \le c\bb I_d$, and so these rank one operators are also bounded. By compactness, we have a subsequence $(n_m)_{m=1}^{\infty}$ so that, $\lim_m R_{i, n_m}=R_i$ which is also rank one. Thus, \begin{align*}
\Psi(X) = \lim_m \Psi_{n_m}(X) = \lim_m \sum_{i=1}^k R_{i,n_m}XR_{i,n_m}^* = \sum_{i=1}^k R_i X R_i^*,
\end{align*} and the result follows.
\end{proof}

Proposition \ref{prop:ebr-lower-semicts} yields the following corollary.

\begin{cor}
Let $0\leq t\leq \frac{1}{d+1}$, and let $\Phi_t = t\cl I_d + (1-t)\Psi_d $. If \begin{align*}
\liminf_{t\to \frac{1}{d+1}^-} \ebr(\Phi_t) \leq d^2,
\end{align*} then Zauner's weak conjecture is true.
\end{cor}

This corollary motivates us to compute $\ebr(\Phi_t)$. We show that in dimensions $d=2$ and $d=3$, $\ebr(\Phi_t)=d^2$ when $0\leq t\leq \frac{1}{d+1}$. In fact, we have the following stronger conjecture than $\ebr(\Phi_t)=d^2$ for all $0\leq t\leq \frac{1}{d+1}$.

\begin{conj}\label{conj:Zauner-continuous}
Let $d\geq 2$. For $0\leq t \leq \frac{1}{d+1}$, consider the channel $\Phi_t = t\cl I_d+(1-t)\Psi_d$. There exists a set of $2d^2$ continuous functions \begin{align*}
x_i,y_i:\left[0,\frac{1}{d+1}\right]\to \bb C^d, \quad \|x_i(t)\|=\|y_i(t)\|=1 \text{ for all } 0\leq t\leq \frac{1}{d+1}, 1\leq i\leq d^2,
\end{align*} such that \begin{align*}
\Phi_t(X) = \frac{1}{d}\sum_{i=1}^{d^2} \left(x_i(t)y_i(t)^* \right) X\left(x_i(t)y_i(t)^* \right)^*, \qquad X\in \bb M_d,
\end{align*} for all $0\leq t\leq \frac{1}{d+1}$.
\end{conj}

We show that for dimensions $d=2$ and $d=3$, we can indeed find such continuous families of vectors. To find them we assume a covariance property and some ad hoc assumptions. Indeed, instead of $2d^2$ continuous functions, now we wish to find two continuous functions, \begin{align*}
x,y:\left[0,\frac{1}{d+1} \right]\to \bb C^d, \quad \|x(t)\| = \|y(t)\| = 1 \text{ for all } 0\leq t\leq \frac{1}{d+1},
\end{align*} such that the map defined by \begin{align*}
\widetilde{\Phi}_t(X) = \frac{1}{d}\sum_{i,j=0}^{d-1} \left(\left(W_{i,j}x(t) \right)\left(W_{i,j}y(t) \right)^*\right) X \left(\left(W_{i,j}x(t) \right)\left(W_{i,j}y(t) \right)^*\right)^*
\end{align*} satisfies $\Phi_t=\widetilde{\Phi}_t$, where of course $\{W_{i,j}:0\leq i,j\leq d-1\}$ are the discrete Weyl matrices, and $0\leq t\leq \frac{1}{d+1}$. To find the functions $x$ and $y$, we solve a set of simultaneous equations obtained by comparing their Choi-matrices: \begin{align*}
\Phi_t(E_{k,l})=\widetilde{\Phi}_t(E_{k,l}), \qquad 0\leq k,l\leq d-1,
\end{align*} where $\{E_{k,l}:0\leq k,l\leq d-1\}$ are the canonical matrix units of $\bb M_d$.

The following two lemmas are the result of elementary calculations.

\begin{lemma}\label{lem:WijEklWij}
Let $\{W_{i,j}:0\leq i,j\leq d-1\}\subseteq \bb M_d$ be the discrete Weyl matrices. If $E_{p,q}\in \bb M_d$ is any canonical matrix unit with $p\neq q$ and $0\leq p,q\leq d-1$, then \begin{align*}
\sum_{j=0}^{d-1}W_{i,j}^*E_{p,q}W_{i,j} = 0.
\end{align*}
\end{lemma}

\begin{lemma}\label{lem:prop-of-Wij}
Let $x,y\in \bb C^d$ be unit vectors and let $\{W_{i,j}:0\leq i,j\leq d-1\}\subseteq \bb M_d$ be the discrete Weyl matrices. Let $\Phi:\bb M_d\to \bb M_d$ be a linear map defined by \begin{align*}
\Phi(X) = \sum_{i,j=0}^{d-1} (W_{i,j}x)(W_{i,j}y)^*X(W_{i,j}y)(W_{i,j}x)^* = \sum_{i,j=0}^{d-1} \inner{W_{i,j}^*XW_{i,j}y}{y}W_{i,j}xx^*W_{i,j}^*.
\end{align*} If $E_{k,l}\in \bb M_d$ is a canonical matrix unit, then, $\Phi(E_{k,k})$ is a diagonal matrix for all $0\leq k\leq d-1$; and if $k\neq l$, then the diagonal entries of $\Phi(E_{k,l})$ are all zero.
\end{lemma}

We now prove Conjecture \ref{conj:Zauner-continuous} for $d=2$ and $d=3$. We remark that there are many possible solutions for the functions $x$ and $y$, but we shall not be occupied with finding all of them.

\begin{theorem}\label{thm:cont-family-dim2}
There exists continuous functions $x,y:\left[0,\frac{1}{3}\right]\to \mathbb{C}^2$ such that $\|x(t)\|=\|y(t)\|=1$, and such that \begin{align*}
\Phi_t(X) = \sum_{i,j=0}^{1} (W_{i,j}x(t))(W_{i,j}y(t))^*X(W_{i,j}y(t))(W_{i,j}x(t))^*, \qquad X\in \bb M_2 
\end{align*} for all $0\leq t\leq \frac{1}{3}$, where $\Phi_t=t\cl I_2+(1-t)\Psi_2$ and $\{W_{i,j}:0\leq i,j\leq 1\}$ are the discrete Weyl matrices.
\end{theorem}

\begin{proof}
Without loss of generality we may assume that \begin{align*}
x = \begin{bmatrix}
r \\ \sqrt{1-r^2}e^{i\theta}
\end{bmatrix}, \qquad y = \begin{bmatrix}
se^{i\theta_1} \\ \sqrt{1-s^2}e^{i\theta_2}
\end{bmatrix},
\end{align*} where $0\leq r,s\leq 1$ and $\theta,\theta_1,\theta_2\in [0,2\pi)$, and where we have suppressed the dependence on $t$.

We set
 \begin{align}\label{eq:values-of-theta}
\theta = \frac{\pi}{4}, \qquad \theta_1 = 0, \qquad \theta_2 = \frac{\pi}{4}.
\end{align} 

We define 
\begin{align}\label{eq:ab(t)}
a = \frac{1}{2}\left(\sqrt{\frac{1+3t}{2}} + \sqrt{\frac{1-t}{2}} \right), \qquad b = \frac{1}{2}\left(\sqrt{\frac{1+3t}{2}} - \sqrt{\frac{1-t}{2}} \right),
\end{align}
which are continuous functions of $t$ for $0 \le t \le 1$. Calculation shows that $(1+a)^2 - b^2$ and $(1-a)^2 - b^2$ are both greater than or equal to 0 for $0 \le t \le 1/3$.

Consequently,
\begin{equation}\label{eq:rs(ab)}
\begin{split}
r &= \frac{1}{2}\left(\sqrt{(1+a)^2-b^2} + \sqrt{(1-a)^2-b^2} \right) \\
s &= \frac{1}{2}\left(\sqrt{(1+a)^2-b^2} - \sqrt{(1-a)^2-b^2} \right)
\end{split}
\end{equation} 
are continuous functions on the interval $0 \le t \le 1/3.$

We leave it to the reader to verify that with these definitions, the functions $x(t)$ and $y(t)$ satisfy the theorem. \end{proof}

While the proof is complete, it does not indicate how we arrived at these solutions. 
If we set \begin{align}\label{eq:ab-and-rs}
a=rs, \qquad b=\sqrt{(1-r^2)(1-s^2)}.
\end{align} With the help of Lemma \ref{lem:prop-of-Wij} and comparing $\widetilde{\Phi}_t(E_{k,l})$ with $\Phi_t(E_{k,l})$, we get the following simultaneous equations: \begin{equation}\label{eq:simult-eq-dim2}
\begin{split}
\frac{1+t}{2} &=  a^2+b^2, \\
t &= 2ab\cos(\theta+\theta_1-\theta_2), \\
0 &= 2ab\cos(\theta-\theta_1+\theta_2).
\end{split}
\end{equation}
Our formulas for $r$ and $s$ came from choosing particular values for the $\theta$'s, solving these equations for $a$ and $b$ and then expressing $r$ and $s$ in terms of $a$ and $b$.

When $t=\frac{1}{3}$, it is straightforward to compute from the above equations that \begin{align*}
x\left(\frac{1}{3}\right) = y\left(\frac{1}{3} \right) = \frac{1}{\sqrt{6}} \begin{bmatrix}
\sqrt{3+\sqrt{3}} \\ e^{\frac{i\pi}{4}}\sqrt{3-\sqrt{3}}
\end{bmatrix},
\end{align*} which is indeed one of the fiducial vectors in dimension two mentioned in Ref. \cite{RBSC}. To get the other fiducial vector mentioned there, if in the above equations we had chosen $\theta=\frac{5\pi}{4}$, then one finds that the solutions of the above equations yield the other fiducial vector in Ref. \cite{RBSC}.

\begin{theorem}\label{thm:cont-family-dim3}
There exists continuous functions $x,y:\left[0,\frac{1}{4}\right]\to \mathbb{C}^3$ such that $\|x(t)\|=\|y(t)\|=1$, and such that \begin{align*}
\Phi_t(X) = \sum_{i,j=0}^{2} (W_{i,j}x(t))(W_{i,j}y(t))^*X(W_{i,j}y(t))(W_{i,j}x(t))^*, \qquad X\in \bb M_3, 
\end{align*} for all $0\leq t\leq \frac{1}{4}$, where $\Phi_t=t\cl I_3+(1-t)\Psi_3$, and $\{W_{i,j}:0\leq i,j\leq 2\}$ are the discrete Weyl matrices.
\end{theorem}

\begin{proof}
Consider the functions $u,v:\left[0,\frac{1}{4} \right]\to \bb C^3$ defined by \begin{align*}
u(t) = \begin{bmatrix}
1 \\ \lambda \\ \lambda
\end{bmatrix}, \qquad v(t) = \begin{bmatrix}
\alpha \\ \beta \\ \beta
\end{bmatrix},
\end{align*} where $\lambda, \alpha, \beta$ are also functions of $t$, but we have suppressed the dependence upon $t$. The function $\alpha$ is given by the positive square root of \begin{align*}
\alpha^2 = \frac{5+4t+4\sqrt{1+7t-8t^2}}{81}.
\end{align*} The function $\lambda$ is a solution of $\lambda(\lambda+1)=\rho$, say $\lambda=\frac{-1+\sqrt{1+4\rho}}{2}$, where $\rho$ is given by \begin{align*}
\rho = \frac{1+2t-\sqrt{1+7t-8t^2}}{1-4t}.
\end{align*} Finally, $\beta$ is given by $\beta=-\alpha(\lambda+1)$. 

Now define $x,y:\left[0,\frac{1}{4} \right]\to \bb C^3$ as \begin{align*}
x(t) = \sqrt{3}\|v(t)\|u(t), \qquad y(t) = \frac{v(t)}{\|v(t)\|}.
\end{align*} Clearly, $\|y(t)\|=1$ and another string of calculations shows: \begin{align*}
\|x(t)\|^2 = 3\alpha^2(1+2\lambda^2)(2\lambda^2+4\lambda+3) = 3\alpha^2(4\rho^2+4\rho+3) = 1.
\end{align*}

Computations show that $x$ and $y$ are continuous on the domain and meet the requirements of the theorem.  \end{proof}

One finds that, \begin{align*}
x\left( \frac{1}{4}\right) = \begin{bmatrix}
\sqrt{\frac{2}{3}} \\ \frac{-1}{\sqrt{6}} \\ \frac{-1}{\sqrt{6}}
\end{bmatrix} = y\left(\frac{1}{4} \right).
\end{align*} This particular fiducial vector indeed matches with \cite{RBSC} when $r_0=\sqrt{\frac{2}{3}}$ and $\theta_1=\theta_2=\pi$.

\section{Mutually unbiased bases}\label{sec:MUB}
The existence problem of the existence of $d+1$ mutually unbiased bases in $\bb C^d$ ($d\geq 2$) is another major open problem in quantum information theory (see Ref. \cite{DEBZ,Kantor}). Earlier we saw that Zauner's weak conjecture for dimension $d$ is equivalent to $\ebr(\ff Z)=d^2$. This leaves open the question of what types of bounds one can obtain on $\ebr(\ff Z)$. In this final section we show that, whenever $d+1$ mutually unbiased bases exist, then $\ebr(\ff Z) \le d(d+1)$. It is known that $d+1$ mutually unbiased bases do exist whenever $d$ is a prime power (see Ref.\cite{Ivan81,WF,BBRV}). So our result holds whenever $d$ is a prime power.

\begin{definition}
Let $\cl V = \{v_1, ..., v_d\}$ and $\cl W=\{w_1,...,w_d\}$ be two orthonormal bases of $\bb C^d$. We say that $\cl V$ and $\cl W$ are \emph{mutually unbiased} if $|\inner{v_i}{w_j}|^2 = \frac{1}{d}$ for all $1\leq i,j\leq d$. (It is easy to see that if $|\inner{v_i}{w_j}|^2=c$ is constant for all $1\leq i,j\leq d$, then $c=\frac{1}{d}$.)
\end{definition}

In what follows, our aim is to show that if there exist $d+1$ mutually unbiased bases $\cl V_i = \{v_{i,j}:1\leq j\leq d\}$ ($1\leq i\leq d+1$) for $\bb C^d$, then the channel defined by \begin{align}\label{eq:channel-from-MUB}
\Phi(X) = \frac{1}{d+1}\sum_{i=1}^{d+1}\sum_{j=1}^d P_{i,j}XP_{i,j}, \qquad X\in \bb M_d
\end{align} where $P_{i,j}$ is the projection onto $\text{span}\{v_{i,j}\}$, is nothing but the channel $\ff Z$.

\begin{lemma}\label{LI-set-MUB}
For $1\leq i\leq d+1$, let $\cl V_i = \{v_{i,j}:1\leq j\leq d\}$ be a collection of $d+1$ mutually unbiased bases of $\bb C^d$. If $P_{i,j}=v_{i,j}v_{i,j}^*$ for all $i,j$, then the set \begin{align*}
\cl P = \left\lbrace P_{1,j}:1\leq j\leq d \right\rbrace \cup  \left\lbrace P_{i,j}: 2\leq i\leq d+1, 1\leq j\leq d-1 \right\rbrace \subseteq \bb M_d,
\end{align*} is a basis for $\bb M_d$.
\end{lemma}

\begin{proof}
Since $|\cl P|=d^2$, it is sufficient to show that it is linearly independent in $\bb M_d$. Using the fact that $\{\cl V_i\}_{i=1}^{d+1}$ are mutually unbiased bases, for $1\leq i,k\leq d+1$ and $1\leq j,l\leq d$, we have \begin{align}\label{eq:MUB-rel}
\tr(P_{i,j}P_{k,l})  = |\inner{v_{i,j}}{v_{k,l}}|^2 = \begin{cases}
0 &\text{ if } i=k \text{ and } j\neq l, \\
1 &\text{ if } i=k \text{ and } j=l, \\
\frac{1}{d} &\text{ if } i\neq k.
\end{cases}
\end{align}

Suppose there exist scalars $\{\lambda_{i,j}:1\leq i\leq d+1, 1\leq j\leq d-1\}\cup \{\lambda_{1,d}\}$ such that \begin{align}\label{lemma-eq1}
\lambda_{1,d}P_{1,d} + \sum_{i=1}^{d+1}\sum_{j=1}^{d-1}\lambda_{i,j}P_{i,j}  = 0.
\end{align} Taking trace of Equation \eqref{lemma-eq1} yields \begin{align}\label{lemma-eq2}
\lambda_{1,d} + \sum_{i=1}^{d+1}\sum_{j=1}^{d-1}\lambda_{i,j}  = 0.
\end{align} Multiplying Equation \eqref{lemma-eq1} by $P_{1,l}$ for $1\leq l\leq d$, taking trace and using Relations \eqref{eq:MUB-rel}, we have \begin{align*}
0 &= \lambda_{1,d}|\inner{v_{1,d}}{v_{1,l}}|^2 + \sum_{i=1}^{d+1}\sum_{j=1}^{d-1} \lambda_{i,j} |\inner{v_{i,j}}{v_{1,l}}|^2   = \lambda_{1,l} + \frac{1}{d}\sum_{i=2}^{d+1} \sum_{j=1}^{d-1}\lambda_{i,j},
\end{align*} which implies \begin{align}\label{eq:lambdaL1}
\lambda_{1,l} = \frac{-1}{d}\sum_{i=2}^{d+1}\sum_{j=1}^{d-1}\lambda_{i,j}.
\end{align} Since the right hand side of Equation \eqref{eq:lambdaL1} is independent of $l$, we conclude that \begin{align}\label{eq:lambda1}
\lambda_{1,1} = \lambda_{1,2} =... =\lambda_{1,d}=:\lambda_1.
\end{align}

Similarly multiplying Equation \eqref{lemma-eq1} by $P_{k,l}$, where $2\leq k\leq d+1$ and $1\leq l\leq d-1$, taking trace, and using Relations \eqref{eq:MUB-rel} we arrive at \begin{align}\label{eq:lambdaKL}
\lambda_{k,l} = \frac{-1}{d}\left(\lambda_{1,d} + \sum_{\substack{i=1 \\ i\neq k}}^{d+1}\sum_{j=1}^{d-1}\lambda_{i,j}  \right).
\end{align} Again since the right hand side of Equation \eqref{eq:lambdaKL} is independent of $l$, we conclude that \begin{align}\label{eq:lambdak}
\lambda_{k,1} = \lambda_{k,2} = ... = \lambda_{k,d-1} =: \lambda_k.
\end{align}

However, for any $2\leq k\leq d+1$, using Equation \eqref{lemma-eq2}, \begin{align*}
\lambda_k = \lambda_{k,l} &= \frac{-1}{d}\left(\lambda_{1,d} + \sum_{\substack{i=1 \\ i\neq k}}^{d+1}\sum_{j=1}^{d-1}\lambda_{i,j}\right) = \frac{-1}{d}\left( \left(\lambda_{1,d} + \sum_{i=1}^{d+1}\sum_{j=1}^{d-1}\lambda_{i,j} \right)  -  \sum_{j=1}^{d-1}\lambda_{k,j}\right) \\
&= \frac{-1}{d}\left( -  \sum_{j=1}^{d-1}\lambda_{k,j}\right) = \frac{d-1}{d}\lambda_k,
\end{align*} which yields $\lambda_k=0$. Then Equation \eqref{lemma-eq1} reduces to \begin{align*}
0 = \lambda_{1,1}P_{1,1} + ... + \lambda_{1,d}P_{1,d} = \lambda_1(P_{1,1}+...+P_{1,d}) = \lambda_1\bb I_d,
\end{align*} and thus $\lambda_1=0$. Hence we have shown that $\lambda_{i,j}=0$ for all $1\leq i\leq d+1$ and $1\leq j\leq d-1$, and $\lambda_{1,d}=0$. \end{proof}

\begin{theorem}\label{MUB-to-Zauner}
For $1\leq i\leq d+1$, let $\cl V_i = \{v_{i,j}:1\leq j\leq d\}$ be a collection of $d+1$ mutually unbiased bases of $\bb C^d$. If $P_{i,j}=v_{i,j}v_{i,j}^*$ for all $i,j$, then the linear map $\Phi:\bb M_d\to \bb M_d$ defined by Equation \eqref{eq:channel-from-MUB} satisfies $\Phi = \ff Z$. Consequently, if $d+1$ mutually unbiased bases exist, then $\ebr(\ff Z) \le d(d+1)$.
\end{theorem}

\begin{proof}
Clearly, the map $\Phi$ is completely positive, and it is straightforward to verify that it is trace preserving and unital. Using Relations \eqref{eq:MUB-rel}, it follows that for $1\leq k\leq d+1$ and $1\leq l\leq d$, \begin{align}\label{eq:Phi-on-Pij}
\Phi(P_{k,l}) &=  \frac{1}{d+1}\sum_{i=1}^{d+1}\sum_{j=1}^{d} |\inner{v_{i,j}}{v_{k,l}}|^2 P_{i,j} = \frac{1}{d+1}\left(P_{k,l} + \bb I_d\right) = \ff Z(P_{k,l}).
\end{align} Since a subset of $\{P_{k,l}:1\leq k\leq d+1,1 \leq k \leq d\}$ forms a basis for $\bb M_d$ as shown in Lemma \ref{LI-set-MUB}, it follows that $\Phi=\ff Z$. \end{proof}

In dimensions for which both a SIC POVM and $d+1$ MUB's exist, the above theorem provides examples of situations in which $\ff Z$ has two distinct Choi-Kraus representations which are not trivially related.

\section{acknowledgements}
We thank the referees for many helpful suggestions to improve the exposition of this article. SKP and JP were supported by Graduate Fellowships at the Department of Pure Mathematics, University of Waterloo. VIP is supported by the Natural Sciences and Engineering Research Council (NSERC) grant number 03784. MR is an Assistant Professor at Birla Institute of Technology and Science (Pilani - Goa Campus).

\end{document}